\documentclass[11pt]{article}

\usepackage[psamsfonts]{eucal}
\usepackage{amsmath}
\usepackage{amsfonts}
\usepackage{amssymb}
\usepackage{amstext}
\usepackage[margin=1in]{geometry}
\usepackage{amscd}
\usepackage{amsthm}
\usepackage{makeidx}
\usepackage{graphicx}
\usepackage{hyperref}
\usepackage{supertabular}
\usepackage{enumerate}
\usepackage{titling}
\usepackage{wasysym}
\usepackage{comment}
\usepackage{cleveref}
\usepackage{listings}

\usepackage{microtype}
\usepackage{color} % red, green, blue, yellow, cyan, magenta, black, white
\usepackage{amsfonts}

\newcommand{\Pro}{\textnormal{Pr}}

\newcommand{\E}{\mathbb{E}}
\newtheorem{thm}{Theorem}

\newtheorem{lemma}[thm]{Lemma}

\numberwithin{thm}{section}

\begin{document}
\title{A New Algorithm for the Robust Semi-random Independent Set Problem}
\author{Theo McKenzie\thanks{University of California, Berkeley.}\and Hermish Mehta\footnotemark[1]\and Luca Trevisan\thanks{Bocconi University. Luca Trevisan was supported by the NSF under grant CCF 1815434 and his work on this project has received funding from the European Research Council (ERC) under the European Union’s Horizon 2020 research and innovation programme (grant agreement No. 834861).}}
\date{}
\maketitle
\begin{abstract}
    We study the independent set problem in a semi-random model proposed by Feige and Kilian. This model selects a graph with a planted independent set of size $k$ and then allows an adversary to modify a large fraction of edges: the subgraph induced by the complement of the independent set can be modified arbitrarily, and the adversary may add (but not delete) edges from the independent set to its complement. In particular, the adversary can create a graph in which the initial planted independent set is not the largest independent set. 
    Feige and Kilian presented a randomized algorithm, which with high probability recovers an independent set of size at least $k$ (which may not be the planted one) when $k=\alpha n$ where $\alpha$ is a constant, and the probability of a random edge $p>(1+\epsilon)\ln n/\alpha n$. Steinhardt studied a restriction of this model in which the adversary is not allowed to add edges from the planted independent set to its complement, and focused on the problem of finding the planted independent set. He develops an algorithm that, given a random ``seed'' vertex in the planted independent set, finds the planted independent set provided that $k= \Omega(n^{2/3} {\log n}^{1/3})$ in the $p=1/2$ regime. Equivalently, by guessing the seed, the algorithm is able to output a list of at most $n$ independent sets of size $k$ such that one of them is the planted one.
    
    We give a new deterministic algorithm in the Feige-Kilian model that finds an independent set of size at least $.99k$ provided that the planted set has size $k=\Omega({n^{2/3}}/{p^{1/3}})$, and finds a list of independent sets, one of which is the planted one provided that $k=\Omega({n^{2/3}}/{p})$. This improves on the algorithm of Feige and Kilian by working for smaller $k$ if $p=\Omega(1/n^{1/3})$, and improves on the algorithm of Steinhardt by working for slightly smaller $k$ and by working against a stronger adversarial model. The ability to find a good approximation of the largest independent set is new when $p<\ln n/k$.
\end{abstract}
\newpage
\section{Introduction}

The maximum independent set problem is, given a graph $G = (V, E)$, find the largest set of mutually non-adjacent vertices.  The associated decision problem, to determine whether a graph contains an independent set of size at least $k$, was one of Karp's twenty-one problems originally proved to be \textsf{NP}-complete  $\cite{KARP}$\footnote{In this work, as well as in others listed below, the planted clique problem is considered. Here we are describing these results in the equivalent point of view in which one wants to find a planted independent set in the complementary graph.}. More recent hardness of approximation results \cite{Zuck} show that, for every $\epsilon >0$, it is impossible to approximate maximum independent set to within $n^{1 - \epsilon}$ in the worst case, unless $\textsf{P} = \textsf{NP}$.

The worst-case hardness of this problem has motivated the study of its
average-case complexity  and of its complexity in semi-random models that are intermediate between average-case analysis and worst-case analysis.

A classical model for the average-case analysis of graph algorithms is the $G_{n,p}$ Erd\'os-R\'enyi model, where each edge is independently present with some probability $p$. In such a model, the largest independent set has size  about $2 \log_b n$ with high probability, for $b = (1-p)^{-1}$ and $n$ the number of vertices in the graph  \cite{Matula}. A simple greedy algorithm finds, with high probability, an independent set of size about $\log_b n$. It has been a long-standing open problem to give an algorithm that finds an independent set of size $(1+\epsilon)\log_b n$.

Another classical model is the ``planted independent set'' one, in which one starts
from a $G_{n,p}$ random graph and then one picks a random set of $k=k(n)$ vertices and removes all existing edges among those vertices, turning them into an independent set. For $p=1/2$, if $k \gg \log n$, the selected set of vertices (which we will call the ``planted'' independent set of the graph) is, with high probability, the unique maximum independent set of the graph. In this case, the problem of finding the largest independent set in the graph coincides with the ``recovery'' problem of identifying the selected set of vertices.

When the size $k(n)$ of the the planted independent set is $\Omega({\sqrt{n \log n}})$, choosing the $k$ vertices of lowest degree is sufficient to find the hidden independent set \cite{Kucera}; Alon, Krivelevich and Sudakov \cite{AKS} give a spectral algorithm to find the planted independent set with high probability when $k(n)=\Omega(\sqrt n)$. It is an open problem whether there is a polynomial time algorithm that finds the planted independent set with high probability in the regime $k(n) = o(\sqrt n)$. Barak, Hopkins, Kelner, Kothari, Moitra, and Potechin established this is impossible for sum of squares algorithms \cite{Barak}.

When studying simple generative models for graphs, such as $G_{n,1/2}$ or planted independent set models, there is a risk of coming up with algorithms that perform well in the model, but that are an overfit for it. For example picking the $k$ vertices of lowest degree is a good way to find a size $k$ independent set in the planted model (if $k\gg \sqrt{n \log n}$) but it may not perform well in practice.

In order to validate the robustness of average-case analyses of algorithms, there has been interest in the study of semi-random generative models, in which a graph is generated via a combination of random choices and adversarial choices. Even though no simple probabilistic model can capture all the subtle properties of realistic graph distributions, realistic distributions can be captured by semi-random models if the way in which the realistic distribution differs from the simple probabilistic model is interpreted as the action of the adversary.

Moreover, by studying semi-random models we gain insight into what part of a problem governs its hardness. If an algorithm in a random graph solves a problem in polynomial time with high probability, then we can ask how adversarial we can make our graph while still solving it in polynomial time with high probability. For example, Feige and Kilian believed the planted independent set should be recoverable without regard for the edges that do not touch the vertices of the independent set, so sought algorithms that could find the maximum independent set when these edges were made adversarial \cite{FK}. In order to gain insight on what instances of unique games can be difficult, Kolla, Makarychev, and Makarychev  created algorithms that solved unique games with high probability in a model where out of 4 given steps of creating a $(1-\epsilon)$ satisfiable instance, only 1 is randomized \cite{KMM}.

Semi-random generative models for graphs were first introduced by Blum and Spencer \cite{BS}, and then further studied by Feige and Kilian \cite{FK}.

In the Feige-Kilian model, one generates a graph with a planted size $k$ independent set as follows: a set $S$ of $k$ vertices is chosen at random; Then, edges from $S$ to $V-S$ are selected as in a $G_{n,p}$ model; finally, an adversary is allowed to choose arbitrarily the edges within vertices in $V-S$, and they are allowed to add edges from $S$ to $V-S$. Note that, when $k<n/2$, the planted set $S$ need not be a largest independent set in the graph since, for example, the adversary could choose to create an independent set of size $k+1$ among the vertices in $V-S$.

Feige and Kilian studied the complexity of finding an independent set {\em of size at least $k$} in the graph arising from their model. They prove that, for $\epsilon >0$, they can solve the problem in polynomial time if 
$p > (1 + \epsilon) \ln n / k$ and $k=\alpha n$ for constant $\alpha$, and, if $p < (1 - \epsilon) \ln n / k$, the problem is not solvable in polynomial time, unless $\textsf{NP} \subseteq \textsf{BPP}$. Since then, progress has been made on weaker monotone semirandom versions of the problem \cite{Chen, CO1, FKr}. Moreover, Coja-Oghlan generalized Feige and Kilian's algorithm to sparse subgraphs \cite{CO2} as opposed to independent sets. However, prior to this paper, there had been no algorithm that improved on the size of the independent set in the Feige-Kilian model.

Steinhardt \cite{JS} studied the recovery problem (that is, the problem of finding the vertices of the planted independent set) in a slight restriction of the Feige-Kilian model with $p=1/2$, in which the adversary can choose edges arbitrarily  within  $V-S$ but cannot add edges between $S$ and $V-S$. Although the problem of recovering $S$ seems to be information-theoretically impossible when $k< n/2$, Steinhardt studies a ``list-decoding'' version of the problem in which the goal is to output a collection of sets, one of which is the planted independent set (or to output the planted independent set given a random vertex sampled from $S$). Steinhardt shows that the problem 
is  information-theoretically unsolvable when $k = o(\sqrt n)$ by showing that a regime in which each vertex belongs to a large number of independent sets has small total variation distance from the semirandom model. He also, along with Charikar and Valiant, gives a polynomial time recovery algorithm when $k = \Omega (n^{2/3} \log^{1/3}n)$ \cite{Char}.

In this paper, we provide a deterministic polynomial time algorithm
that finds an independent set of size $\geq (1-\epsilon) \cdot k$ in the Feige-Kilian model (in which the adversary is
allowed to add edges between $S$ and $V-S$) that works with high probability for $k = \Omega( n^{2/3}/p^{1/3})$, that is, $k = 
\Omega(n/ d^{1/3})$, where $d=pn$ is the average degree. Moreover, we can find a list of independent sets one of which is the planted one, if $k = \Omega( n^{2/3}/p)$. Comparing this to Feige and Kilian's model, this works for a wider range of the parameter $k$, which can be sublinear, but for a narrower range of the parameter $p$:  our full-recovery algorithm must have $p=\Omega(1/{n^{1/3}})$, whereas the Feige-Kilian algorithm only requires $p>(1+\epsilon)/(\alpha n)$. However, our $1-\epsilon$ factor approximation can work for $p=\Omega(1/n)$. 

Instantiating our algorithm with $p=1/2$, we have an improvement over Steinhardt's algorithm, by gaining a $\log^{1/3}$ factor in the size parameter $k$ for which the algorithm works, and by being able to deal with the full Feige-Kilian adversarial model.

Specifically, our results are the following:

\begin{thm}
\label{thm: approx}
 For all $\epsilon>0$, there exists a constant  $c_1(\epsilon)$ such that when
\[k\geq c_1\frac {n^{2/3}}{p^{1/3}}\]
we can, with high probability,  return an independent set of size at least $(1-\epsilon)k$.
\end{thm}

\begin{thm}
\label{thm: main} There exists a constant $c_2$ such that 
when \[k \geq c_2\frac{n^{2/3}}{p}\]we can, with high probability,  return at most $n$ candidate solutions such that one is the original independent set. In particular, we are able
to find an independent set of size at least $k$.
\end{thm}

While Steinhardt relied on spectral techniques, we use semidefinite programming (SDP). The improved robustness comes from the robustness of the SDP technique, and the logarithmic gain comes from an analysis of the SDP via the Grothendieck inequality, which tends to give tighter information about the properties of random graphs than spectral bounds obtained from matrix Chernoff bounds or related techniques.

A natural way to apply SDP techniques in this setting would be to solve an SDP relaxation of the maximum independent set problem on the given graph. This, however, would not work, because the adversary could create a large set with few edges in $V-S$, and the optimum of the relaxation could be related to this other set and carry no information about $S$.

Instead, we use a ``crude SDP'' (C-SDP), a technique used by Kolla, Makarychev and Makarychev in their work on semi-random Unique Games \cite{KMM} and by Makarychev, Makarychev and Vijayaraghavan \cite{MMV} in their later work on semi-random cut problems. Crude SDPs are not relaxations of the problem of interest and, in particular, there is no standard way of mapping an intended solution (in our case, the set $S$) to an associated canonical feasible solution of the SDP. Rather, the crude SDP is designed in such a way that the optimal solution reveals information about the planted solution.

Our crude SDP will associate a unit vector to each vertex, with the constraint that adjacent vertices are mapped to orthogonal vectors; the goal of the SDP is to maximize the sum of inner products among all pairs of vectors. The point of the analysis will be that, with high probability over the choice of the graph, and for every possible choice of the adversary, the optimal solution of the SDP will map the vertices in $S$ to vectors that are fairly close to one another, and then $S$ can be recovered by looking for sets of $k$ vertices whose associated vectors are clustered. Specifically, our program will be 

\begin{equation*}
    \begin{array}{cc}
        \text{maximize} & \sum_{u,v} \langle x_u,x_v\rangle \\
        \text{subject to}&\\
        &||x_u||^2 = 1,\forall u\in V\\
        &\langle x_u, x_v \rangle = 0, \forall (u,v) \in E 
\end{array}
\end{equation*}

Note that, because of the constraints $||x_u||^2 = 1$, the cost function $\max \sum_{u,v} \langle x_u, x_v \rangle$ is equivalent to
$\min \sum_{u,v} || x_u - x_v ||^2$.

For motivation for why this may work, consider the following relaxation of the maximum independent set problem, which is a formulation of the Lov\'asz theta function. This function can retrieve the independent set in more restricted planted models \cite{FKr}.

\begin{equation*}
    \begin{array}{cc}
        \text{maximize} & \sum_{u,v}\langle x_u, x_v\rangle\\
        \text{subject to} &\\
        & \sum_{u} ||x_u||^2 = 1\\
                 &\langle x_u, x_v \rangle = 0,\forall (u,v) \in E
    \end{array}
\end{equation*}

%[Luca: from here on rewrite, see email]
The only difference in requirements is that instead of requiring the norms of the vectors to sum to 1, we require that the norm of each vector is 1. 

To gain intuition about the difference between the two SDPs, it is useful to see that, even if the goal is to recover an independent set of size at least $\Omega(k)$, even in the regime, say $k=n^{3/4}$, the solution found by theta function will not necessarily be helpful in the Feige-Kilian model.
Suppose for example that, in $G$, we have all possible edges between $S$ and $V-S$, and that the subgraph induced by $V-S$ is a graph 
in which the maximum independent set has size $\leq n^{.1}$ but the value of the theta function is $\geq n^{.99}$ (such a graph
is constructed by Feige \cite{F}). Note that an adversary can construct such a graph in the Feige-Kilian model. The optimum of the
theta function will set all vectors corresponding to vertices in $S$ to zero, and the remaining vertices to the optimal solution
of the theta function for the graph induced by $V-S$. This is because all vertices in $S$ must be associated with vectors that are orthogonal to all vectors associated with vertices in $V-S$, so any feasible solution is a convex combination of a solution entirely concentrated on $S$ and a solution entirely concentrated in $V-S$: any solution concentrated in $S$ can have cost at most $k$ while solutions concentrated in $V-S$ can have cost $\geq n^{.99}$, so the optimal solution will put zero weight on vertices in $S$. But given a solution entirely concentrated in $V-S$, no rounding algorithm can find a good solution, since all independent sets in $V-S$ have a size $\leq n^{.1}$.

Our C-SDP can be seen as a relaxation of the following problem: find a coloring of $G$ that maximizes the sum of squares of the sizes of the color classes, that is, find a partition of the vertices $\mathcal P=\{P_1\ldots P_m\}$ such that each $P_i$ is an independent set and, subject to this constraint, maximize $\sum_{i=1}^m |P_i|^2$. To see that our C-SDP is a relaxation of this problem, take any partition $\mathcal P=\{P_1\ldots P_m\}$, choose $m$ orthogonal unit vectors $x_1,\ldots,x_m$, and then associate to each vertex in $P_i$ the vector $x_i$. This will satisfy the orthogonality constraints of the C-SDP, and the cost function will evaluate to $\sum_i |P_i|^2$. 

Even in the presence of other independent sets larger than the planted one, we may hope that it is optimal 
for the above combinatorial problem to have the planted independent set be a color class, and that it is optimal
for the relaxation to associate to the vertices of the planted independent set a set of vectors that are close to each other,
that is, such that
\[\sum_{\substack{u,v\\ u\in S \textnormal{ or } v\in S}}\langle x_u^*,x_v^*\rangle\]
will be approximately $k^2$, where $x_u^*$ is the vector corresponding to $u\in V$ in the optimal solution of the C-SDP.

We prove this via an argument by contradiction: if a solution does not cluster the vectors corresponding to the vertices in $S$ close together, then we can construct a new feasible solution of lower cost, meaning that the original solution was not optimal. To bound the cost of the new solution we need to understand the sum of distances-squared, according to the original solution, between pairs of vertices in $S \times (V-S)$. This is where we use the Grothendieck inequality: to reduce this question to a purely combinatorial question that can be easily solved using the expansion of the connection between $S$ and $V\backslash S$ and union bounds.

The argument that we use to prove that an optimal solution to C-SDP must cluster the vertices of the planted independent set, because otherwise a feasible solution of larger cost would exist, relies on the assumption that $k = \Omega(n^{2/3}/p^{1/3})$. Although our analysis stops working when $k = o(n^{2/3}/p^{1/3})$, it is not clear whether the algorithm stops working as well, that is, whether there is a strategy for the adversary
in the $k = o(n^{2/3}/p^{1/3})$ regime which produces graphs that, with high probability, have C-SDP optimal solutions that map the vertices of the planted independent set to nearly orthogonal vectors such that $\sum_{u,v \in S} \langle x_u , x_v \rangle = o(k^2)$, although we suspect that this is the case.

Feige and Kilian were motivated to study robust algorithms for the planted independent set problem  because of how it relates to the $( n/k)$-coloring problem. In the semirandom $(n/k)$-coloring problem, we have $n/k$ planted independent sets each of size $k$ and add edges at random between these sets. Then, we can add any other edges that keep the planted independent sets independent.

Finding an algorithm for the maximum independent set gives us an algorithm to solve this planted $(n/k)$-coloring problem, as in this model, with high probability the largest independent sets will be the planted ones. It is worth noting that our improved independent set algorithm can also solve the $(n/k)$-coloring problem for $k=\Omega(n^{2/3}/p)$. Since Feige and Kilian's original paper, algorithms specifically constructed for the coloring problem can solve it with high probability when $k=\Omega(\sqrt{n\log n/p})$ \cite{CO1}.

\section{C-SDP Clustering}

An independent set $S$ of a graph $G=(V,E)$ is a subset $S \subseteq V$ such that the subgraph induced by $S$ does not contain any edges.  We form our semi-random graph $G$ as follows, using the same formulation as Feige and Kilian $\cite{FK}$. Here $\overline S = V- S$.

\begin{enumerate}
    \item An adversary chooses a set $S \subseteq V$ such that $|S| = k$.
    
    \item Create a graph $G' = (V, E')$, where each pair of vertices $u, v$ forms an edge independently with probability $P(u,v)$. $P(u,v)$ is formulated as follows.

    \[
    P(u, v) =
    \begin{cases}
        0 &: (u,v) \in S \times S \\
        p &: (u,v) \in S \times \overline S \\
    0 &: (u,v) \in \overline S\times \overline S
    \end{cases}
    \]

    \item The adversary can add any edge $(u,v)$ arbitrarily as long as $(u,v) \notin S\times S$. Our graph will be of the form $G=(V,E)$ where $S$ is an independent set in $G$ and $E\supset E'$. 
\end{enumerate}

This gives us a graph that is arbitrary on $(u,v) \in \overline S \times \overline S$, has no edges within $S$, and is lower bounded by Bernoulli random variables on the boundary $S \times \overline S$.

Our goal for this section will be to show the following:
\begin{lemma}
\label{lemma: closeness}
Call $x_u^*$ the vector corresponding to $u \in V$ for the optimal solution to our C-SDP. With high probability,
\[
\sum_{u,v\in S}||x_u^*-x_v^*||^2=O\left(\frac{n\sqrt k}{\sqrt p}\right)
\]
\end{lemma}

Our first step will be to show the following bound:

\begin{lemma}
\label{lemma: crude}
For the optimal solution we have
\[
\sum_{u,v\in S\times S}||x_u^*-x_v^*||^2+2\sum_{u,v\in S\times \overline S}||x_u^*-x_v^*||^2\leq 4k(n-k)
\]
\end{lemma}

\begin{proof}
Here our analysis will be simpler if we think of our cost function as the equivalent
\[
\textnormal{minimize } \sum_{u,v} ||x_u^*-x_v^*||^2
\]

This is equivalent as all of our vectors are norm 1, so

\[\sum_{u,v} ||x_u-x_v||^2=\sum_{u,v} ||x_u||^2-2\langle x_u,x_v\rangle +||x_v||^2=\sum_{u,v} 2-2\langle x_u,x_v\rangle.\]

Consider the feasible solution to the SDP obtained by taking the optimal solution, then setting all vectors corresponding to $u \in S$ to a single unit vector $e$ orthogonal to all other vectors. We keep all vectors corresponding to vertices in $\overline S$ as the same as the optimal solution. Call $x_u'$ the vector for our new adjusted solution corresponding to vertex $u$. Since $x^*$ is optimal we have

\begin{eqnarray*}
    \left( \sum_{u, v \in V\times V} ||x_u^*-x_v^*||^2 - ||x_u'-x_v'||^2 \right) \leq 0.
    \end{eqnarray*}
    Terms that do not contain a vertex corresponding to $S$ do not change between the two solutions. Therefore we cancel these out, which yields
    \begin{eqnarray*}
    \left( \sum_{u, v \in S \times S} ||x_u^*-x_v^*||^2 - ||x_u'-x_v'||^2 \right) + 2\left( \sum_{u, v \in S \times \overline S} ||x_u^*-x_v^*||^2 - ||x_u'-x_v'||^2 \right) \leq0.
    \end{eqnarray*}
    We then know that for our adjusted solution, $||x_u'-x_v'||^2=0$ if both $u$ and $v$ are in $S$, and $||x_u'-x_v'||^2=2$ if only one of $u,v$ is in $S$, giving
    \begin{eqnarray*}
    \left(\sum_{u, v \in S\times S} ||x_u^*-x_v^*||^2 \right) + 2\left(\sum_{u, v \in S \times \overline S} ||x_u^*-x_v^*||^2 \right) - 4k(n-k) \leq0&\\
    \left( \sum_{u, v \in S \times S} ||x_u^*-x_v^*||^2 \right) + 2\left( \sum_{u, v \in S \times \overline S} ||x_u^*-x_v^*||^2 \right) \leq 4k(n-k)&
\end{eqnarray*}
\end{proof}
Our next step is to show that the second sum in Lemma~\ref{lemma: crude} is large. Towards this end we show the following.

\begin{lemma}
\label{lemma: Groth}
With high probability, for the initial random choice of edges $E'$,
\[
\max_{\substack{x_1,\ldots, x_n\\||x_u||=1, \\ \forall u \in V}}\left|\sum_{u\in S, v\in \overline S}\langle x_u,x_v\rangle-\frac1p\sum_{\substack{(u,v)\in E'\\u\in S,v\in\overline S}}\langle x_u,x_v\rangle  \right|= O\left(\frac{n\sqrt k}{\sqrt p}\right).
\]

\end{lemma}
\begin{proof}
We proceed along the lines of \cite{GW}. Set
\begin{eqnarray*}
D &:=&\max_{\substack{x_1,\ldots, x_n\\||x_u||=1, \\ \forall u \in V}}\left|\sum_{u\in S,v\in \overline S}\langle x_u,x_v\rangle-\frac1p\sum_{\substack{(u,v)\in E'\\u\in S,v\in\overline S}}\langle x_u,x_v\rangle\right|\\
&=& \frac1p\max_{\substack{x_1,\ldots, x_n\\||x_u||=1, \\ \forall u \in V}}\left|\sum_{u\in S,v\in \overline S}(p-A_{uv})\langle x_u,x_v\rangle\right|\\
\end{eqnarray*}
where $A_{uv}$ is the entry of the adjacency matrix corresponding to the vertex pair $u,v$ and the edge set $E'$.

We define a new matrix $M$ such that
\[
M_{uv}=\left\{
\begin{array}{cc}
      p-A_{uv} &(u,v)\in S\times \overline S\\
   0&\textnormal{otherwise}
\end{array}\right.
\]

We have
\[
\sum_{u,v\in V}M_{uv}\langle x_u,x_v\rangle=\sum_{u\in S, v\in \overline S}(p-A_{uv})\langle x_u,x_v\rangle
\]

We then have that there exists a constant $c$ such that
\[\max_{\substack{x_1,\ldots, x_n\\||x_u||=1\\\forall u\in V}}\left|\sum_{u,v\in V}M_{uv}\langle x_u,x_v\rangle\right|\leq c\max_{\substack{x_1,\ldots,x_n\in\{\pm1\}^n\\y_1,\ldots,y_n\in\{\pm 1\}^n}} \left|\sum_{u,v\in V}M_{uv}x_uy_v\right|\] by Grothendieck's inequality (\cite{GR}, see for example \cite{GrothIntro}). For a fixed set of $x_1,\ldots,x_n,y_1,\ldots, y_n\in \{\pm 1\}^{2n}$, 
\[
\sum_{u,v\in V}M_{uv}x_uy_v=\sum_{u
\in S,v\in \overline S}(p-A_{uv})x_uy_v.
\]
Each entry $A_{uv}$ corresponding to $u\in S,v\in \overline S$ is a Bernoulli random 0-1 variable. Each term $(p-A_{uv})x_uy_v$ has absolute value at most 1 and expectation 0.
Therefore, by Bernstein's inequality,
\begin{eqnarray}
\Pro\left[\left|\sum_{u\in S,v\in \overline S} (p-A_{uv})x_uy_v\right|\geq\epsilon p k(n-k)\right]<2e^{-\epsilon^2 p k(n-k)/3}
\end{eqnarray}
for any $0\leq\epsilon\leq 1$.

There are $2^{2n}$ possibilities for assignments of $x_u,y_v$, giving us
\begin{flalign}
\Pro\Bigg[\max_{\substack{x_1,\ldots,x_n=\{\pm1\}^n\\y_1,\ldots,y_n=\{\pm 1\}^n}}\left|\sum_{u\in S,v\in \overline S} (p-A_{uv})x_uy_v\right|\geq\epsilon pk(n-k)\Bigg]
&\leq 2^{2n+1}e^{-\epsilon^2 p k(n-k)/3}\\&<e^{1.5n-\epsilon^2pk(n-k)/3}
\end{flalign}
We have (2) from (1) and union bounding over assignments of $x_uy_v$. (3) is true if $n>6$.  

Therefore\[\Pro \left[D<c\epsilon k(n-k)\right]> 1-e^{1.5n-\epsilon^2pk(n-k)/3}.\]

If we set $\epsilon=\frac{3}{\sqrt{pk}}$ then with high probability

\[
\max_{\substack{x_1,\ldots, x_n\\||x_u||=1, \\ \forall u \in V}}\left|\sum_{u\in S, v\in \overline S}\langle x_u,x_v\rangle-\frac1p\sum_{\substack{(u,v)\in E'\\u\in S,v\in\overline S}}\langle x_u,x_v\rangle  \right|= O\left(\frac{n\sqrt k}{\sqrt p}\right).\]

\end{proof}

\begin{proof}[Proof of Lemma~\ref{lemma: closeness}]

By Lemma~\ref{lemma: crude},
\begin{eqnarray*}
\sum_{u,v\in S\times S}||x_u^*-x_v^*||^2+2\sum_{u,v\in S\times \overline S}||x_u^*-x_v^*||^2&\leq& 4k(n-k)\\
\sum_{u,v\in S\times S}||x_u^*-x_v^*||^2+2\left(\sum_{u,v\in S\times \overline S}2-2\langle x_u^*,x_v^*\rangle \right)&\leq& 4k(n-k).
\end{eqnarray*}
By pulling out the constant term, we get
\begin{eqnarray*}
\sum_{u,v\in S\times S}||x_u^*-x_v^*||^2-4\sum_{u,v\in S\times \overline S} \langle x_u^*,x_v^*\rangle &\leq& 0.
\end{eqnarray*}
Then, by Lemma~\ref{lemma: Groth}
\begin{eqnarray*}
\sum_{u,v\in S\times S}||x_u^*-x_v^*||^2-\frac4p\sum_{(u,v)\in E'}\langle x_u^*,x_v^*\rangle -O\left(\frac{n\sqrt{k}}{\sqrt p}\right)&\leq& 0.\\
\end{eqnarray*}
By the requirements of the SDP, $\langle x_u^*,x_v^*\rangle=0$ for all $(u,v)\in E'$. Therefore 

\[
\sum_{u,v\in S\times S}||x_u^*-x_v^*||^2= O\left(\frac{n\sqrt{k}}{\sqrt p}\right).
\]
\end{proof}
Note that this argument works even when the adversary adds edges to the boundary of $S$, as the vertex pairs corresponding to $E'$ will approximate $p$ of the overall sum with high probability regardless of whether the other vertex pairs correspond to edges or not. Our argument only requires that the vertex pairs corresponding to $E'$ correspond to edges, meaning the adversary cannot remove edges from the boundary, but they can add them.

\section{Algorithm Analysis and Recovery}

Our algorithm is as follows:
\begin{itemize}
    \item 
    Solve the crude SDP.
    \item
    For each vector, create a set of all vectors that are $\ell_2$ distance less than $\sqrt{2-\sqrt 2}$ from the original vector. Namely we take the ball of radius $\sqrt{2-\sqrt{2}}$ around the the vector $x_u$ and list all vectors inside the ball. We call this set of vertices $S_u$.
    \item
    Add to the set all vertices that are independent with all vertices already in the set.
    \item
    Return the largest such set.
\end{itemize}
We use this algorithm to prove our two theorems.
\begin{proof}[Proof of Theorem~\ref{thm: approx}]
Using Lemma \ref{lemma: closeness}, we have that with high probability there exists a constant $c_3$ such that

\[
\sum_{u,v\in S}||x_u^*-x_v^*||^2\leq c_3 \frac{n\sqrt k}{\sqrt p}
\]
for the optimal solution of our SDP. Therefore there is some vertex $u\in S$ such that

\[
\E_{v\in S}(||x_u^*-x_v^*||^2) \leq c_3 \frac{n}{k^{3/2}\sqrt p}.
\]

By Markov's inequality we have

\[
\Pro_{v\in S}(||x_u^*-x_v^*||^2\geq 2-\sqrt 2)\leq c_4 \frac{n}{k^{3/2}\sqrt p}.
\]
where $c_4=c_3/(2-\sqrt 2)$. We choose $\sqrt{2-\sqrt 2}$ because if $||x_u-x_v||^2<2-\sqrt 2$ then as $||x_u-x_v||^2=2-2\langle x_u,x_v\rangle$, this implies that $\cos(x_u,x_v)=\langle x_u,x_v\rangle>1/\sqrt 2$. Therefore any two points within the ball of radius $2-\sqrt 2$ around $x_v$ cannot be orthogonal, so the ball of this radius forms an independent set.

This means that after running the SDP,
\begin{eqnarray}
|S_u|\geq k-c_4\frac{n}{\sqrt {kp}}.
\end{eqnarray}
Thereby 
 for a given $\epsilon >0$, if we set $c_1(\epsilon)=(c_4/\epsilon)^{2/3}$, then if
\[k\geq c_1\frac {n^{2/3}}{p^{1/3}}\]
$|S_u|\geq (1-\epsilon)k$. As $S_u$ will be one of our $n$ candidate solutions, we can, with high probability, discover an independent set of size at least $(1-\epsilon)k$.
\end{proof}

This means that for $p$ as small as $p=\Theta(\frac1n)$ our algorithm will discover a large independent set, assuming large enough $k$. This is surprising, considering that in this regime, with high probability there are many vertices in $\overline S$ with no edges to $S$.

\begin{proof}[Proof of Theorem~\ref{thm: main}]
After adding vertices to $S_u$ greedily, we hope that the resulting set will be $S$. For this to be the case, we need to make sure that all vertices of $\overline S$ share at least one edge with $S_u$. This is at most the probability that every vertex of $\overline S$ has at least $k-|S_u|$ edges to the set $S$. Namely

\begin{eqnarray}
\min_{v\in \overline S} \deg(v,S)> k-|S_u|
\end{eqnarray}

By (4), this will be satisfied if

\begin{eqnarray}
\min_{v\in \overline S} \deg(v,S)>  c_4\frac{n}{\sqrt{kp}}
\end{eqnarray}

To satisfy (6) for the minimum ${v\in\overline S}$, it is of course necessary to satisfy (6) in the expectation over vertices in $\overline S$. Therefore we need

\[
kp> c_4\frac{n}{\sqrt{kp}}.
\]
Therefore, let us require that
\[
k\geq \frac{2c_4^{2/3}n^{2/3}}{p}.
\]
To show that this requirement is in fact sufficient, we have by Chernoff bounds that 

\begin{eqnarray}
    \Pr(\exists v \in \overline{S} \textnormal{ such that }\deg(v,S)\leq (1-\epsilon)kp) < (n - k) e^{-\epsilon^2kp/2}<ne^{-\epsilon^2c_4^{2/3} n^{2/3}}
\end{eqnarray}
so for $\epsilon=\frac12$, this probability will go to 0. Therefore (5) occurs with high probability.

We find, with high probability, all vertices in $\overline{S}$ share at least one edge with the elements from $S_u$. Hence, no vertices outside $S$ will be included in the set corresponding to $u$, as orthogonal vectors are $\ell_2$ distance $\sqrt{2}$ away. Since only elements from $S$ will be present in $S_u$, the remaining vertices of $S$ will be added during the greedy step. Therefore, when the algorithm terminates, this set will contain the original planted independent set with high probability.
\end{proof}

This means that our algorithm can give full recovery as long as $p=\Omega(\frac1{n^{1/3}})$ for $k$ a constant fraction of $n$. We do not hope for full recovery for $p<(1-\epsilon)\ln n/k$, as Feige and Kilian showed that with high probability this is impossible unless $\textsf{NP} \subseteq \textsf{BPP}$.

If we are given a vertex at random such as in the model of \cite{JS}, then we can recover the original set exactly.

\begin{thm}
If we are given a random vertex of $S$, then with high probability, we can recover the set $S$ when $k\geq c_2 {n^{2/3}}/p$.
\end{thm}

\begin{proof}
We add the following steps to the algorithm:
\begin{itemize}
    \item 
    Remove all independent sets of size less than $c_2n^{2/3}/p$ from our list.
    \item
    For sets $S_1,S_2$ on the list, if $S_1\neq S_2$ and $|S_1\cap S_2|\geq (1-p/2) |S_1|$ then remove $S_1$.
    \item
    For sets $S_1, S_2$ on the list, if $S_1\neq S_2$ and $|S_1\cap S_2|\geq 3(\log n)/p$, remove both $S_1$ and $S_2$ from the list.
    \item
    If our random vertex $u\in S$ is in exactly 1 set on our list, return this set. Otherwise, return FAIL.
\end{itemize}

First we will show that with high probability $S$ remains on the list by the end of the algorithm. If $S'$ is on the list before the first removal step, it is necessarily maximal by the greedy step. Therefore if $S'\neq S$, then $\exists v\in S'$ such that $v\notin S$. For $S'$ to be an independent set, there can be no edges from $v$ to $S'\cap S$. By (7), with high probability, every vertex in $\overline S$ satisfies $\deg(v,S)>kp/2$, meaning $|S'\cap S|<(1-p/2)k$, and $S$ survives this step. 

For $S'\neq S$ on the list immediately after the second removal step, we have $|S'\cap S|< (1-p/2)|S'|$ so $|S'\cap \overline S|>p|S'|/2>c_2n^{2/3}/2>3(\log n)/p$ for large enough $n$, as we must have $p=\Omega(1/n^{1/3})$. If $|S'\cap S|\geq 3(\log n)/p$, then there is $T\subset S'$ such that $|T\cap S|=3(\log n)/p$ and $|T\cap\overline S|=3(\log n)/p$.

For any such set $T$, the probability that $T$ is an independent set is at most $(1-p)^{(3(\log n)/p)^2}$. The probability that an independent $T$ exists is
\begin{eqnarray}
\Pr(\exists \textnormal{ independent }T)&\leq& \binom k{3(\log n)/p}\binom{n-k}{3(\log n)/p} (1-p)^{(3(\log n)/p)^2}\\
&<&k^{3(\log n)/p}(n-k)^{3(\log n)/p}e^{-p(3(\log n)/p)^2}\\
&=&k^{3(\log n)/p}(n-k)^{3(\log n)/p}n^{-9(\log n)/p}\\
&=&\left(\frac{k(n-k)}{n^3}\right)^{3(\log n)/p}
\ll 1
\end{eqnarray}
meaning that $S$ survives the second removal with high probability. (9) follows from the inequality $1-x\leq e^{-x}$.

If $s$ represents the number of unique independent sets on the list at the end of the algorithm, then by the inclusion exclusion principle,
\[
s\frac{c_2n^{2/3}}p-\binom s2\frac{3\log n}p \leq n
\]
as there are $n$ elements overall. Therefore

\[
sc_2n^{2/3}-\binom s23\log n \leq np.
\]

 From this we can see that we must have
\[
s\leq \frac{2n^{1/3}}{c_2}
\]
for large enough $n$.
The number of elements of $S$ that will appear in other independent sets in our list is at most \[s\frac{3\log n}p\leq \frac{6n^{1/3}\log n}{c_2p}.\] Therefore the probability that the random vertex of $S$ we receive is in any of the other sets remaining is at most \[\frac{6n^{1/3}\log n}{c_2pk}\leq\frac{6\log n}{c_2^2n^{1/3}}\ll 1.\]
\end{proof}

\end{document}